\documentclass[final]{siamltex}

\usepackage{epsfig}
\usepackage{array}
\usepackage{amssymb}

\newtheorem{remark}[theorem]{Remark}
\newtheorem{assumption}[theorem]{Assumption}

\title{Causal Rate Distortion Function and Relations to Filtering Theory}

\author{Photios A. Stavrou\thanks{Ph.D. student at ECE Department, University of Cyprus, Green Park, Aglantzias 91,
P.O. Box 20537, 1687, Nicosia, Cyprus ({\tt photios.stavrou@ucy.ac.cy}).}
        \and Charalambos D. Charalambous\thanks{Professor at ECE Department, University of Cyprus, Green Park, Aglantzias 91, P.O. Box 20537, 1687, Nicosia, Cyprus ({\tt chadcha@ucy.ac.cy}).}}

\begin{document}

\maketitle

\begin{abstract}
A causal rate distortion function (RDF) is defined, existence of extremum solution is described via weak$^*$-convergence, and its relation to filtering theory is discussed. The relation to filtering is obtained via a causal constraint imposed on the reconstruction kernel to be realizable while the extremum solution is given for the stationary case.
\end{abstract}



\pagestyle{myheadings}
\thispagestyle{plain}
\markboth{P. A. STAVROU AND C. D. CHARALAMBOUS}{CAUSAL RATE DISTORTION FUNCTION AND FILTERING THEORY}

\section{Introduction}\label{introduction} Shannon's information theory for reliable communication evo-lved over the years without much emphasis on real-time realizability or causality imposed on the communication sub-systems. In particular, the classical rate distortion function (RDF) for source data compression deals with the characterization of the optimal reconstruction conditional  distribution subject to a fidelity criterion \cite{berger}, without regard for realizability. Hence, coding schemes which achieve the RDF are not realizable.\\
On the other hand, filtering theory is developed by imposing real-time realizability on estimators with respect to measurement data.  Although, both reliable communication and filtering (state estimation for control) are concerned with reconstruction of processes, the main underlying assumptions characterizing them are different.
\par In this paper, the intersection of rate distortion function (RDF) and realizable filtering theory is established by invoking the additional assumption that the reconstruction kernel is realizable via causal operations, while the optimal causal reconstruction kernel is derived. Consequently, the connection between causal RDF, its characterization via the optimal reconstruction kernel, and realizable filtering theory are established under very general conditions on the source (including Markov sources). The fundamental advantage of the new filtering approach based on causal RDF, is the ability to ensure average or probabilistic bounds on the estimation error, which is a non-trivial task when dealing with Bayesian filtering techniques.
\par The first relation between information theory and filtering via distortion rate function is discussed by R. S. Bucy in \cite{bucy}, by carrying out the computation of a realizable distortion rate function with square criteria for two samples of the Ornstein-Uhlenbeck process. The earlier work of A. K. Gorbunov and M. S. Pinsker \cite{gorbunov91} on $\epsilon$-entropy defined via a causal constraint on the reproduction distribution of the RDF, although not directly related to the realizability question pursued by  Bucy, computes the causal RDF for stationary Gaussian processes via power spectral densities. The realizability constraints imposed on the reproduction conditional distribution in \cite{bucy} and \cite{gorbunov91} are different. The actual computation of the distortion rate or RDF in these works is based on the Gaussianity of the process, while no general theory is developed to handle arbitrary processes.\\
The main results described are the following.
\begin{enumerate}
\item[1)] Existence of the causal RDF using the topology of weak$^*$-convergence.
\item[2)] Closed form expression of the optimal reconstruction conditional distribution for stationary processes, which is realizable via causal operations.
\item[3)] Realization procedure of the filter based on the causal RDF.
\end{enumerate}
Next,we give a high level discussion on Bayesian filtering theory and we present some aspects of the problem and results pursued in this paper.
Consider a discrete-time process $X^n\triangleq\{X_0,X_1,\ldots,X_n\}\in{\cal X}_{0,n} \triangleq \times_{i=0}^n{\cal X}_i$, and its reconstruction $Y^n\triangleq\{Y_0,Y_1,\ldots,Y_n\}\in{\cal Y}_{0,n} \triangleq \times_{i=0}^n{\cal Y}_i$, where ${\cal X}_i$ and ${\cal Y}_i$ are Polish spaces (complete separable metric spaces). The objective is to reconstruct $X^n$ by $Y^n$ causally subject to a distortion or fidelity criterion.
\par In classical filtering, one is given a mathematical model that generates the process $X^n$, $\{P_{X_i|X^{i-1}}(dx_i|x^{i-1}):i=0,1,\ldots,n\}$ often induced via discrete-time recursive dynamics, a mathematical model that generates observed data obtained from sensors, say, $Z^n$, $\{P_{Z_i|Z^{i-1},X^i}$ $(dz_i|z^{i-1},x^i):i=0,1,\ldots,n\}$ while $Y^n$ are the causal estimates of some function of the process $X^n$ based on the observed data $Z^n$. Thus, in classical filtering theory both models which generate the unobserved and observed processes, $X^n$ and $Z^n$, respectively, are given \'a priori. Fig.~\ref{filtering} illustrates the cascade block diagram of the filtering problem.	
\begin{figure}[ht]
\centering
\includegraphics[scale=0.60]{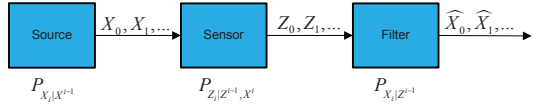}
\caption{Block Diagram of Filtering Problem}
\label{filtering}
\end{figure}

\par In causal rate distortion theory one is given the process $X^n$, which induces $\{P_{X_i|X^{i-1}}(dx_i|x^{i-1}):~i=0,1,\ldots,n\}$, and determines the causal reconstruction conditional distribution $\{P_{Y_i|Y^{i-1},X^i}(dy_i|y^{i-1},x^i):~i=0,1,\ldots,n\}$ which minimizes the mutual information between $X^n$ and $Y^n$ subject to a distortion or fidelity constraint, via a causal (realizability) constraint. The filter $\{Y_i:~i=0,1,\ldots,n\}$ of $\{X_i:~i=0,1,\ldots,n\}$ is found by realizing the reconstruction distribution $\{P_{Y_i|Y^{i-1},X^i}(dy_i|y^{i-1}$,\\
$x^i):~i=0,1,\ldots,n\}$ via a cascade of sub-systems as shown in Fig.~\ref{filtering_and_causal}.
\begin{figure}[ht]
\centering
\includegraphics[scale=0.60]{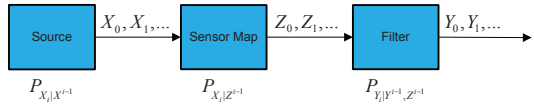}
\caption{Block Diagram of Filtering via Causal Rate Distortion Function}
\label{filtering_and_causal}
\end{figure}

The distortion function or fidelity constraint between  $x^n$ and its reconstruction $y^n$, is a measurable function defined by
 \begin{eqnarray}
 d_{0,n} : {\cal X}_{0,n} \times {\cal Y}_{0,n} \mapsto [0, \infty], \: \: d_{0,n}(x^n,y^n)\triangleq\sum^n_{i=0}\rho_{0,i}(x^i,y^i)\nonumber
 \end{eqnarray}
The mutual information between $X^n$ and $Y^n$, for a given distribution ${P}_{X^n}(dx^n)$, and  conditional distribution $P_{Y^n|X^n}(dy^n|x^n)$, is defined by
\begin{equation}
{I}(X^n;Y^n)\triangleq \int_{{\cal X}_{0,n}\times{\cal Y}_{0,n}}\log\Big(\frac{P_{Y^n|X^n}(dy^n|x^n)}{{P}_{Y^n}(dy^n)}\Big)
P_{Y^n|X^n}(dy^n|x^n)\otimes{P}_{X^n}(dx^n) \label{1}
\end{equation}
Define the $(n+1)-$fold causal convolution measure
\begin{equation}
{\overrightarrow P}_{Y^n|X^n}(dy^n|x^n) \triangleq \otimes^n_{i=0}P_{Y_i|Y^{i-1},X^i}(dy_i|y^{i-1},x^i)-a.s. \label{9}
\end{equation}
The realizability constraint for a causal filter is defined by
\begin{eqnarray}
{\overrightarrow Q}_{ad} \triangleq\Big\{ P_{Y^n|X^n}(dy^n|x^n) :~P_{Y^n|X^n}(dy^n|x^n) ={\overrightarrow P}_{Y^n|X^n}(dy^n|x^n)-a.s. \Big\}  \label{eq18}
\end{eqnarray}
The realizability condition (\ref{eq18}) is necessary, otherwise the connection between filtering and realizable rate distortion theory cannot be established. This is due to the fact that $P_{Y^n|X^n}(dy^n|x^n)=\otimes_{i=0}^n{P}_{Y_i|Y^{i-1},X^n}(dy_i|y^{i-1},x^n)-a.s.$, and hence in general, for each $i=0,1, \ldots,n$, the conditional distribution of $Y_i$ depends on future symbols $\{X_{i+1},X_{i+2},\ldots,X_n\}$ in addition to the past and present symbols $\{Y^{i-1},X^i\}$.\\
{\it Causal Rate Distortion Function.}
The causal RDF is defined by
\begin{equation}
{R}^c_{0,n}(D)\triangleq \inf_{{ P}_{Y^n|X^n}(dy^n|x^n)\in  \overrightarrow{Q}_{ad}:~{E}\big\{d_{0,n}(X^n,Y^n)\leq{D}\big\}}I(X^n; Y^n)\label{7}
\end{equation}
Note that realizability condition (\ref{eq18}) is different from the realizability condition  in \cite{bucy}, which is defined under the assumption that $Y_i$ is independent of $X_{j|i}^*\triangleq X_j -\mathbb{E}\Big(X_j|X^i\Big), j=i+1, i+2, \ldots,$. The claim here is that realizability condition (\ref{eq18}) is more natural and applies to processes which are not necessarily Gaussian having square error distortion function. Realizability condition (\ref{eq18}) is weaker than the causality condition in \cite{gorbunov91} defined by  $X_{n+1}^\infty \leftrightarrow X^n \leftrightarrow Y^n$ forms a Markov chain.\\
 The point to be made regarding  (\ref{7}) is that  (see also Lemma~\ref{lem1}):
\begin{eqnarray}
&{P}_{Y^n|X^n}(dy^n|x^n)= {\overrightarrow P}_{Y^n|X^n}(dy^n|x^n)-a.s.
{\Longleftrightarrow} \nonumber\\
&I(X^n;Y^n)=\int \log\Big(\frac{{\overrightarrow P}_{Y^n|X^n}(dy^n|x^n)}{{P}_{Y^n}(dy^n)}\Big){\overrightarrow P}_{Y^n|X^n}(dy^n|x^n){P}_{X^n}(dx^n)
\equiv{\mathbb I}(P_{X^n},{\overrightarrow P}_{Y^n|X^n})\label{eq6}
\end{eqnarray}
where ${\mathbb I}(P_{X^n},{\overrightarrow P}_{Y^n|X^n})$ points out the functional dependence of $I(X^n;{Y^n})$ on $\{P_{X^n}$,\\
${\overrightarrow P}_{Y^n|X^n}\}$.\\
The paper is organized as follows. Section~\ref{problem_formulation} discusses the formulation on abstract spaces. Section~\ref{existence} establishes  existence of optimal minimizing kernel, and Section~\ref{necessary} derives the stationary solution. Section~\ref{realization1} describes the realization of causal RDF. Throughout the manuscript proofs are omitted due to space limitation.

\section{Problem Formulation}\label{problem_formulation}

Let $\mathbb{N}^n\triangleq\{0,1,\ldots,n\}$, $n \in \mathbb{N} \triangleq \{0,1,2,\ldots\}$. The source and reconstruction alphabets, respectively, are sequences of Polish spaces $\{ {\cal X}_t: t\in\mathbb{N}\}$ and $\{ {\cal Y}_t: t\in\mathbb{N}\}$, associated with their corresponding measurable spaces $({\cal X}_t,{\cal B}({\cal X}_t))$ and $({\cal Y}_t, {\cal B}({\cal Y}_t))$, $t\in\mathbb{N}$. Sequences of alphabets are  identified
with the product spaces $({\cal X}_{0,n},{\cal B}({\cal X}_{0,n})) \triangleq {\times}_{k=0}^{n}({\cal X}_k,{\cal B}({\cal X}_k))$,
and $({\cal Y}_{0,n},{\cal B}({\cal Y}_{0,n}))\triangleq \times_{k=0}^{n}({\cal Y}_k,{\cal B}({\cal Y}_k))$.
The source and reconstruction are processes denoted by $X^n \triangleq \{X_t: t\in\mathbb{N}^n\}$, $X:\mathbb{N}^{n}\times\Omega\mapsto {\cal X}_t$, and by $Y^n\triangleq \{Y_t: t\in\mathbb{N}^n\}$, $Y:\mathbb{N}^{n}\times\Omega\mapsto  {\cal Y}_t$, respectively. Probability measures on any measurable space $( {\cal Z}, {\cal B}({\cal Z}))$ are denoted by ${\cal M}_1({\cal Z})$. It is assumed  that the $\sigma$-algebras $\sigma\{X^{-1}\}=\sigma\{Y^{-1}\}=\{\emptyset,\Omega\}$.

\begin{definition}\label{stochastic kernel}
Let $({\cal X}, {\cal B}({\cal X})), ({\cal Y}, {\cal B}({\cal Y}))$ be measurable spaces in which $\cal Y$ is a Polish Space. A  stochastic kernel on $\cal Y$ given $\cal X$ is a mapping $q: {\cal B}({\cal Y}) \times {\cal X}  \rightarrow [0,1]$ satisfying the following two properties:
\par 1) For every $x \in {\cal X}$, the set function $q(\cdot;x)$ is a probability measure (possibly finitely additive) on ${\cal B}({\cal Y}).$
\par 2) For every $F \in {\cal B}({\cal Y})$, the function $q(F;\cdot)$ is ${\cal B}({\cal X})$-measurable.\\
 The set of all such stochastic Kernels is denoted by ${\cal Q}({\cal Y};{\cal X})$.
\end{definition}
\begin{definition}\label{comprchan}
Given measurable spaces $({\cal X}_{0,n},{\cal B}({\cal X}_{0,n}))$, $({\cal Y}_{0,n},{\cal B}({\cal Y}_{0,n}))$, then
\par 1) {\it A Non-Causal Data Compression Channel} is a  stochastic kernel $ q_{0,n} (dy^n; x^n) \in {\cal Q}({\cal Y}_{0,n} ;{\cal X}_{0,n})$ which admits a factorization into a non-causal sequence
\begin{eqnarray}
q_{0,n}(dy^n; x^n)=\otimes_{i=0}^n q_i(dy_i;y^{i-1},x^n)\nonumber
\end{eqnarray}
where  $q_i(dy_i;y^{i-1},x^n) \in {\cal Q}({\cal Y}_i;{\cal Y}_{0,i-1}\times{\cal X}_{0,n}), i=0,\ldots,n,~n \in \mathbb{N}$.
\par 2) {\it A Causally Restricted Data Compression Channel} is a stochastic kernel $q_{0,n} (dy^n$\\
$;x^n) \in {\cal Q}({\cal Y}_{0,n} ;{\cal X}_{0,n})$ which admits a factorization into a causal sequence
\begin{eqnarray}
q_{0,n}(dy^n; x^n)=\otimes_{i=0}^n q_i(dy_i;y^{i-1},x^i)-a.s.,\nonumber
\end{eqnarray}
where $q_i \in {\cal Q}({\cal Y}_i;{\cal Y}_{0,i-1}\times{\cal X}_{0,i}), i=0,\ldots,n,~n \in \mathbb{N}$.
\end{definition}
\subsection{Causal Rate Distortion Function}
In this subsection the causal RDF is defined.
Given a source probability measure ${\cal \mu}_{0,n} \in {\cal M}_1({\cal X}_{0, n})$ (possibly finite additive) and a reconstruction Kernel $q_{0,n} \in {\cal Q}({\cal Y}_{0, n};{\cal X}_{0, n})$, one can define three probability measures as follows.
\par (P1): The joint measure $P_{0,n} \in {\cal M}_1({\cal Y}_{0,n}\times {\cal X}_{0, n})$:
\begin{eqnarray}
P_{0,n}(G_{0,n})&\triangleq&(\mu_{0,n} \otimes q_{0,n})(G_{0,n}),\:G_{0,n} \in {\cal B}({\cal X}_{0,n})\times{\cal B}({\cal Y}_{0,n})\nonumber\\
&=&\int_{{\cal X}_{0,n}} q_{0,n}(G_{0,n,x^n};x^n) \mu_{0,n}(d{x^n})\nonumber
\end{eqnarray}
where $G_{0,n,x^n}$ is the $x^n-$section of $G_{0,n}$ at point ${x^n}$ defined by $G_{0,n,x^n}\triangleq \{y^n \in {\cal Y}_{0,n}: (x^n, y^n) \in G_{0,n}\}$ and $\otimes$ denotes the convolution.
\par (P2): The marginal measure $\nu_{0,n} \in {\cal M}_1({\cal Y}_{0,n})$:
\begin{eqnarray}
\nu_{0,n}(F_{0,n})&\triangleq& P_{0,n}({\cal X}_{0, n} \times F_{0,n}),~F_{0,n} \in {\cal B}({\cal Y}_{0,n})\nonumber\\
&=&\int_{{\cal X}_{0, n}} q_{0,n}(({\cal X}_{0, n}\times F_{0,n})_{{x}^{n}};{x}^{n}) \mu_{0,n}(d{x^n})=\int_{{\cal X}_{0, n}} q_{0,n}(F_{0,n};x^n) \mu_{0,n}(dx^n)\nonumber
\end{eqnarray}
\par(P3): The product measure  $\pi_{0,n}:{\cal B}({\cal X}_{0,n}) \times
{\cal B}({\cal Y}_{0,n}) \mapsto [0,1] $ of $\mu_{0,n}\in{\cal M}_1({\cal X}_{0, n})$ and $\nu_{0,n}\in{\cal M}_1({\cal Y}_{0, n})$ for $G_{0,n} \in {\cal B}({\cal X}_{0,n}) \times {\cal B}({\cal Y}_{0,n})$:
\begin{eqnarray}
\pi_{0,n}(G_{0,n})\triangleq(\mu_{0,n} \times \nu_{0,n})(G_{0,n})=\int_{{\cal X}_{0, n}} \nu_{0,n}(G_{0,n,x^n}) \mu_{0,n}(dx^n)\nonumber
\end{eqnarray}
The precise definition of mutual information between two sequences of Random Variables $X^n$ and $Y^n$, denoted $I(X^n; Y^n)$ is defined via the Kullback-Leibler distance (or relative entropy) between the joint probability distribution of  $(X^n, Y^n)$ and the product of its marginal probability distributions of $X^n$ and $Y^n$, using the Radon-Nikodym derivative. Hence, by the chain rule of relative entropy:
\begin{eqnarray}
I(X^n;Y^n) &\triangleq&  \mathbb{D}(P_{0,n}|| \pi_{0,n})
=\int_{{\cal X}_{0,n} \times {\cal Y}_{0,n}}\log \Big( \frac{d  (\mu_{0,n} \otimes q_{0,n}) }{d ( \mu_{0,n} \times \nu_{0,n} ) }\Big) d(\mu_{0,n} \otimes q_{0,n}) \nonumber\\
& =& \int_{{\cal X}_{0,n} \times {\cal Y}_{0,n}} \log \Big( \frac{q_{0,n}(d y^n; x^n)}{  \nu_{0,n} (dy^n)   } \Big)q_{0,n}(dy^n;dx^n)\otimes\mu_{0,n}(dx^n) \nonumber\\
&=&\int_{{\cal X}_{0,n}} \mathbb{D}(q_{0,n}(\cdot;x^n)|| \nu_{0,n}(\cdot)) \mu_{0,n}(dx^n)\equiv \mathbb{I}(\mu_{0,n}, q_{0,n})  \label{re3}
\end{eqnarray}

\par The next lemma relates causal product reconstruction kernels and conditional independence.
\begin{lemma} \label{lem1}
The following are equivalent for each $n\in\mathbb{N}$.
\begin{enumerate}
\item[1)] $q_{0,n} (dy^n; x^n)={\overrightarrow q}_{0,n}(dy^n;x^n)$-a.s., defined in Definition \ref{comprchan}-2).

\item[2)] For each $i=0,1,\ldots, n-1$,  $Y_i \leftrightarrow (X^i, Y^{i-1}) \leftrightarrow (X_{i+1}, X_{i+2}, \ldots, X_n)$, forms a Markov chain.



\item[3)] For each  $i=0,1,\ldots, n-1$, $Y^i \leftrightarrow X^i \leftrightarrow X_{i+1}$ forms a Markov chain.
\end{enumerate}
\end{lemma}
According to Lemma~\ref{lem1}, for causally restricted kernels
\begin{eqnarray}
I(X^n;Y^n)&=&\int_{{\cal X}_{0,n} \times {\cal Y}_{0,n}} \log \Big( \frac{ \overrightarrow{q}_{0,n}(d y^n; x^n)}{\nu_{0,n}(dy^n)} \Big){\overrightarrow q}_{0,n}(dy^n;dx^n)\otimes\mu_{0,n}(dx^n) \nonumber  \\
&\equiv&{\mathbb I}(\mu_{0,n},\overrightarrow{q}_{0,n} )  \label{ex11}
\end{eqnarray}
where (\ref{ex11}) states that $I(X^n;Y^n)$ is a functional of $\{\mu_{0,n},{\overrightarrow q}_{0,n}\}$.
Hence, causal RDF is defined by optimizing ${\mathbb I}(\mu_{0,n},{q}_{0,n})$ over ${q}_{0,n}{\in}Q_{0,n}(D)$ where $Q_{0,n}(D)=\{ q_{0,n} \in {\cal Q}({\cal Y}_{0,n}; {\cal X}_{0,n}):\int_{{\cal X}_{0,n}}\int_{ {\cal Y}_{0,n}} d_{0,n}(x^n,y^n) q_{0,n}(dy^n;x^n)\otimes\mu_{0,n}(dx^n) \leq D\}$ subject to the realizability constraint $q_{0,n}(dy^n;x^n)={\overrightarrow q}_{0,n}(dy^n;x^n)-a.s.,$ which satisfies a distortion constraint, or via (\ref{ex11}).
\begin{definition}\label{def1}
$($Causal Rate Distortion Function$)$
Suppose $d_{0,n}(x^n,y^n)\triangleq\sum^n_{i=0}\rho_{0,i}(x^i,y^i)$, where $\rho_{0,i}: {\cal X}_{0,i}  \times {\cal Y}_{0,i}\rightarrow [0, \infty)$, is a sequence of ${\cal B}({\cal X}_{0,i}) \times {\cal B }( {\cal Y}_{0,i})$-measurable distortion functions, and let $\overrightarrow{Q}_{0,n}(D)$ (assuming is non-empty) denotes the average distortion or fidelity constraint defined by
\begin{eqnarray}
\overrightarrow{Q}_{0,n}(D)\triangleq Q_{0,n}(D)\bigcap{\overrightarrow Q}_{ad},~D\geq0\nonumber
\end{eqnarray}\\
The causal RDF associated with the causally restricted kernel is defined by
\begin{eqnarray}
{R}^c_{0,n}(D) \triangleq  \inf_{{{q}_{0,n}\in \overrightarrow{Q}_{0,n}(D)}}{\mathbb I}(\mu_{0,n},{q}_{0,n})\label{ex12}
\end{eqnarray}
\end{definition}
\section{Existence of Optimal Causal Reconstruction Kernel}\label{existence}
\par In this section, appropriate topologies and function spaces are introduced and existence of the minimizing causal product kernel in $(\ref{ex12})$ is shown. 
\subsection{Abstract Spaces}
Let $BC({\cal Y}_{0,n})$ denote the vector space of bounded continuous real valued functions defined
on the Polish space ${\cal Y}_{0,n}$. Furnished with the sup norm
topology, this is a Banach space. The topological dual of $BC({\cal Y}_{0,n})$ denoted by $ \Big( BC({\cal Y}_{0,n})\Big)^*$ is isometrically isomorphic to the Banach space of finitely additive regular bounded signed measures on ${\cal Y}_{0,n}$ \cite{dunford1988}, denoted by $M_{rba}({\cal Y}_{0,n})$. Let $\Pi_{rba}({\cal Y}_{0,n})\subset M_{rba}({\cal Y}_{0,n})$ denote the set of regular bounded
finitely additive probability measures on ${\cal Y}_{0,n}$.  Clearly if ${\cal Y}_{0,n}$ is compact,
then $\Big(BC({\cal Y}_{0,n})\Big)^*$ will be isometrically isomorphic to the space of countably additive signed measures, as
in \cite{csiszar74}. Denote by $L_1(\mu_{0,n}, BC({\cal Y}_{0,n}))$ the space of all $\mu_{0,n}$-integrable functions defined on  ${\cal X}_{0,n}$ with values in $BC({\cal Y}_{0,n}),$ so that for each $\phi \in L_1(\mu_{0,n}, BC({\cal Y}_{0,n}))$ its norm is defined by
\begin{eqnarray}
\parallel \phi \parallel_{\mu_{0,n}} \triangleq \int_{{\cal X}_{0,n}} ||\phi(x^n,\cdot)||_{BC({\cal Y}_{0,n})} \mu_{0,n}(dx^n) <\infty\nonumber
\end{eqnarray}
The norm topology $\parallel{\phi}\parallel_{\mu_{0,n}}$, makes $L_1(\mu_{0,n}, BC({\cal Y}_{0,n}))$ a Banach
space, and it follows from the theory of ``lifting" \cite{tulcea1969} that the dual
of this space is $L_{\infty}^w(\mu_{0,n}, M_{rba}({\cal Y}_{0,n}))$, denoting
the space of all $M_{rba}({\cal Y}_{0,n})$ valued functions $\{q\}$ which
are weak$^*$-measurable in the sense that for each $\phi \in
BC({\cal Y}_{0,n}),$  $x^n \rightarrow q_{x^n}(\phi) \triangleq \int_{{\cal Y}_{0,n}}\phi(y^n) q(dy^n;x^n)$ is $\mu_{0,n}$-measurable and $\mu_{0,n}$-essentially
bounded.
\subsection{Weak$^*$-Compactness and Existence}
Define an admissible set of stochastic kernels associated with classical RDF by
\begin{eqnarray}
Q_{ad}\triangleq L_{\infty}^w(\mu_{0,n}, \Pi_{rba}({\cal Y}_{0,n})) \subset L_{\infty}^w(\mu_{0,n}, M_{rba}({\cal Y}_{0,n}))\nonumber
\end{eqnarray}
Clearly, $Q_{ad}$ is a unit sphere in $L_{\infty}^w(\mu_{0,n}, M_{rba}({\cal Y}_{0,n}))$. For each $\phi{\in}L_1(\mu_{0,n}, BC({\cal Y}_{0,n}))$ we can define a linear functional on $L_{\infty}^w(\mu_{0,n}, M_{rba}({\cal Y}_{0,n}))$ by
\begin{eqnarray}
\ell_{\phi}(q_{0,n})\triangleq\int_{{\cal X}_{0,n}}\Big( \int_{{\cal Y}_{0,n}} \phi(x^n,y^n)q_{0,n}(dy^n;x^n) \Big)\mu_{0,n}(dx^n)\nonumber
\end{eqnarray}
This is a bounded, linear and weak$^*$-continuous functional on $L_{\infty}^w(\mu_{0,n}, M_{rba}({\cal Y}_{0,n}))$.

For $d_{0,n}: {\cal X}_{0,n}  \times {\cal Y}_{0,n}\mapsto[0, \infty)$ measurable and $d_{0,n}{\in}L_1(\mu_{0,n},BC({\cal Y}_{0,n}))$ the distortion constraint set of the classical RDF is $Q_{0,n}(D)\triangleq\{q{\in}Q_{ad}:\ell_{d_{0,n}}(q_{0,n}){\leq}D\}$.
\begin{lemma}\label{$Q_ad$ w$^*$-closed}
For $\ell_{d_{0,n}}{\in}L_1(\mu_{0,n},BC({\cal Y}_{0,n}))$, the set $Q_{0,n}(D)$ is weak$^*$-bounded and weak$^*$-closed subset of $Q_{ad}$.
\end{lemma}

Hence $Q_{0,n}(D)$ is weak$^*$-compact (compactness of $Q_{ad}$ follows from Alaoglu's Theorem~\cite {dunford1988}).
\begin{lemma}\label{weakstar-compact_classL1}
 Let  ${\cal X}_{0,n},{\cal Y}_{0,n}$ be two Polish spaces  and $d_{0,n} :{\cal X}_{0,n}\times{\cal Y}_{0,n}\mapsto[0,\infty]$, a measurable, nonnegative, extended real valued
function, such that for a fixed $x^n \in {\cal X}_{0,n}$, $y^n \rightarrow d(x^n,\cdot)$ is continuous on ${\cal Y}_{0,n}$, for $\mu_{0,n}$-almost all $x^n \in {\cal X}_{0,n}$, and $d_{0,n}\in L_1(\mu_{0,n}, BC({\cal Y}_{0,n}))$. For any $D \in [0,\infty)$, introduce the set
\begin{eqnarray}
{Q}_{0,n}(D)\triangleq\{ q_{0,n}\in {Q}_{ad} : \int_{{\cal X}_{0,n}} \biggl(\int_{{\cal Y}_{0,n}}d_{0,n}(x^n,y^n) {q}_{0,n}(dy^n;x^n)\biggr)\mu_{0,n}(dx^n)\leq D\}\nonumber
\end{eqnarray}
and  suppose it  is nonempty.\\
Then ${Q}_{0,n}(D)$  is a weak$^*$-closed subset of $Q_{ad}$ and hence weak$^*$-compact.
\end{lemma}

Next, we define the realizability constraint via causally restricted kernels as follows
\begin{eqnarray}
{\overrightarrow Q}_{ad}=\Big\{q_{0,n}\in {Q_{ad}}:q_{0,n}(dy^n;x^n)={\overrightarrow q}_{0,n}(dy^n;x^n)-a.s.\Big\}\nonumber
\end{eqnarray}
which satisfy an average distortion function as follows:
\begin{eqnarray}
{\overrightarrow Q_{{0,n}}(D)}&\triangleq&Q_{0,n}(D)\bigcap{\overrightarrow Q}_{ad}\nonumber\\
&=&\Big\{{q}_{0,n} \in {\overrightarrow Q}_{ad} :\ell_{d_{0,n}}({\overrightarrow q}_{0,n})\triangleq \int_{{\cal X}_{0,n}} \biggr(\int_{{\cal Y}_{0,n}}d_{0,n}(x^n,y^n){\overrightarrow q}_{0,n}(dy^n;x^n) \biggr)\nonumber\\
&\otimes&\mu_{0,n}(dx^n)\Big\}\nonumber
\end{eqnarray}
The following is assumed.
\begin{assumption}\label{weakstar-closed}
Let ${\cal X}_{0,n}$ and ${\cal Y}_{0,n}$ be Polish spaces and $\overrightarrow{Q}_{ad}$  weak$^*$-closed.
\end{assumption}
\begin{remark}
The conditions 1) ${\cal Y}_{0,n}$ is a compact Polish space, and 2)  for all $h(\cdot){\in}BC({\cal Y}_{n})$, the function $(x^{n},y^{n-1})\in{\cal X}_{0,n}\times{\cal Y}_{0,n-1}\mapsto\int_{{\cal Y}_n}h(y)q_n(dy;y^{n-1},x^n)\in\mathbb{R}$
is continuous jointly in the variables $(x^{n},y^{n-1})\in{\cal X}_{0,n}\times{\cal Y}_{0,n-1}$ are sufficient for ${\overrightarrow{Q}}_{ad}$ to be weak$^*$-closed.
\end{remark}
\begin{theorem}\label{weakstar-compact_2}
Suppose Assumption~\ref{weakstar-closed} and the conditions of Lemma~\ref{weakstar-compact_classL1} hold. For any $D \in [0,\infty)$, introduce the set
\begin{eqnarray}
{Q}_{0,n}(D)\triangleq\{ q_{0,n}\in {\overrightarrow Q}_{ad} : \int_{{\cal X}_{0,n}} \biggl(\int_{{\cal Y}_{0,n}}d(x^n,y^n) {\overrightarrow q}_{0,n}(dy^n;x^n)\biggr)\mu_{0,n}(dx^n)\leq D\}\nonumber
\end{eqnarray}
and  suppose it  is nonempty.\\
Then ${\overrightarrow Q}_{0,n}(D)$  is a weak$^*$-closed subset of ${\overrightarrow Q}_{ad}$ and hence weak$^*$-compact.
\end{theorem}
\begin{theorem} \label{th3}
Under Theorem~\ref{weakstar-compact_2}, $R^c_{0,n}(D)$ has a minimum.
\end{theorem}
\begin{proof}
 Follows from weak$^*$-compactness of $\overrightarrow{Q}_{ad}$ and lower semicontinuity of $\mathbb{I}(\mu_{0,n},q_{0,n})$ with respect to $q_{0,n}$ for a fixed $\mu_{0,n}$.\qquad
\end{proof}

\section{Necessary Conditions of Optimality of Causal Rate Distortion Function}\label{necessary}

In this section the form of the optimal causal product reconstruction kernels is derived under a stationarity assumption. The method is based on calculus of variations on the space of measures \cite{dluenberger69}.
\begin{assumption}\label{stationarity}
The family of measures $\overrightarrow{q}_{0,n}(dy^n;x^n)=\otimes^n_{i=0}q_i(dy_i;y^{i-1},x^i)-a.s.$, is the convolution of stationary conditional distributions.
\end{assumption}

Assumption~\ref{stationarity} holds for stationary process $\{(X_i,Y_i):i\in\mathbb{N}\}$ and $\rho_{0,i}(x^i,y^i)\equiv\rho(T^i{x^n},T^i{y^n})$, where $T^i{x^n}$ is the shift operator on $x^n$.  Utilizing Assumption~\ref{stationarity}, which holds for stationary processes  and a single letter distortion function, the Gateaux differential of $\mathbb{I}(\mu_{0,n},{q}_{0,n})$ is done in only one direction $\big{(}$since $q_i(dy_i;y^{i-1},x^i)$ are stationary$\big{)}$.
\par The constrained problem defined by (\ref{ex12}) can be reformulated using Lagrange multipliers as follows (equivalence of constrained and unconstrained problems follows similarly as in \cite{dluenberger69}).
\begin{equation}
{R}_{0,n}^c(D) = \inf_{{q}_{0,n} \in {\overrightarrow Q}_{ad}} \Big\{{{\mathbb I}}(\mu_{0,n},{q}_{0,n})-s(\ell_{{d}_{0,n}}({q}_{0,n})-D)\Big\} \label{ex13}
\end{equation}
and  $s \in(-\infty,0]$ is the Lagrange multiplier.\\
Note that ${\overrightarrow Q}_{ad}$ is a proper subset of the vector space $L_{\infty}^w(\mu_{0,n},M_{rba}({\cal Y}_{0,n}))$ which represent the realizability constraint. Therefore, one should introduce another set of Lagrange multipliers to obtain an optimization on the vector space $L_{\infty}^w(\mu_{0,n},M_{rba}({\cal Y}_{0,n}))$ without constraints.
\begin{theorem} \label{th6}
Suppose $d_{0,n}(x^n,y^n)=\sum_{i=0}^n\rho(T^i{x^n},T^i{y^n})$ and the Assumption~\ref{weakstar-closed} holds. The infimum in $(\ref{ex13})$ is attained at  ${q}^*_{0,n} \in
L_{\infty}^w(\mu_{0,n},{\Pi}_{rba}({\cal Y}_{0,n}))\cap{\overrightarrow Q}_{ad}$ given by
\begin{eqnarray}
{q}_{0,n}^*(dy^n;x^n)&=&\overrightarrow{q}^*_{0,n}(dy^n;x^n)-a.s.\nonumber\\
&=&\otimes_{i=0}^nq_i^*(dy_i;y^{i-1},x^i)-a.s\label{ex14}\\
&=&\otimes_{i=0}^n\frac{e^{s \rho(T^i{x^n},T^i{y^n})}\nu^*_i(dy_i;y^{i-1})}{\int_{{\cal Y}_i} e^{s \rho(T^i{x^n},T^i{y^n})} \nu^*_i(dy_i;y^{i-1})},~s\leq{0}\nonumber
\end{eqnarray}
and $\nu^*_i(dy_i;y^{i-1})\in {\cal Q}({\cal Y}_i;{\cal Y}_{0,{i-1}})$. The causal RDF is given by
\begin{eqnarray}
{R}_{0,n}^c(D)&&=sD -\sum_{i=0}^n\int_{{{\cal X}_{0,i}}\times{{\cal Y}_{0,i-1}}}\log \left( \int_{{\cal Y}_i} e^{s\rho(T^i{x^n},T^i{y^n})} \nu^*_i(dy_i;y^{i-1})\bigg)\right.\nonumber\\[-1.5ex]\label{ex15}\\[-1.5ex]
&&\quad\left.\times{{\overrightarrow q}^*_{0,i-1}}(dy^{i-1};x^{i-1})\otimes\mu_{0,i}(dx^i)\right.\nonumber
\end{eqnarray}
If ${R}_{0,n}^c(D) > 0$ then $ s < 0$  and
\begin{eqnarray}
\sum_{i=0}^n\int_{{\cal X}_{0,i}} \int_{{\cal Y}_{0,i}}
\rho(T^i{x^n},T^i{y^n}){\overrightarrow q}^*_{0,i}(dy^i;x^i) \mu_{0,i}(dx^i)=D\nonumber
\end{eqnarray}
\end{theorem}
\\
\begin{remark}
Note that if the distortion function satisfies $\rho(T^i{x^n},T^i{y^n})=\rho(x_i,T^i{y^n})$ then ${q}^*_{i}(dy_i;y^{i-1},x^i)=q_i^*(dy_i;y^{i-1},x_i)-a.s.,~i\in{\mathbb{N}^n}$, that is, the reconstruction kernel is Markov in $X^n$.
\end{remark}

\section{Realization of Causal Rate Distortion Function}\label{realization1}

 Fig.~\ref{realization} illustrates a cascade of sub-systems which realizes the causal RDF. This is called source-channel matching in information theory \cite{gastpar2003}.  It is also described in \cite{charalambous2008} and \cite{tatikonda2000} and is essential in control applications since this technique allows us to design encoding/decoding schemes without delays.
\begin{figure}[ht]
\centering
\includegraphics[scale=0.60]{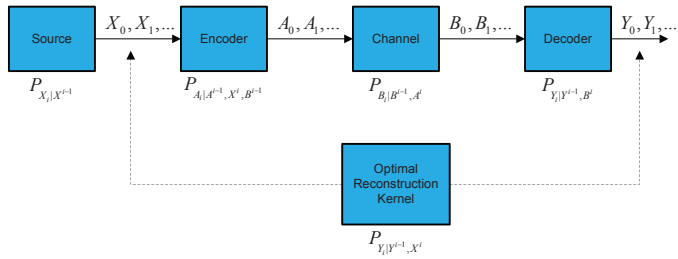}
\caption{Block Diagram of Realizable Causal Rate Distortion Function}
\label{realization}
\end{figure}

Examples to illustrate the concepts can be found in \cite{charalambous2008,stavrou-charalambous-kourtellaris2012b}.

\bibliographystyle{siam}
\bibliography{photis_references}

\end{document}